\documentclass[journal,twoside,web]{ieeecolor}
\usepackage{generic}
\usepackage{cite}
\usepackage{amsmath,amssymb,amsfonts}

\usepackage{amsthm}
\usepackage{algorithmic}
\usepackage{graphicx}
\usepackage{textcomp}
\usepackage{booktabs}
\usepackage{float}
\usepackage{bm}
\usepackage{array}
\usepackage{csquotes}
\usepackage[italicdiff]{physics}
\let\labelindent\relax
\usepackage{enumitem}
\usepackage{stfloats}
\usepackage[italicdiff]{physics}

\newtheorem{theorem}{Theorem}

\def\BibTeX{{\rm B\kern-.05em{\sc i\kern-.025em b}\kern-.08em
    T\kern-.1667em\lower.7ex\hbox{E}\kern-.125emX}}
\begin{document}
\title{A Thickness Sensitive Vessel Extraction Framework for Retinal and Conjunctival Vascular Tortuosity Analysis}
\author{Ashwin De Silva, \IEEEmembership{Member, IEEE}, Malsha V. Perera, \IEEEmembership{Member, IEEE}, Navodini Wijethilake, \IEEEmembership{Member, IEEE}, Saroj Jayasinghe, Nuwan D. Nanayakkara, \IEEEmembership{Member, IEEE} and Anjula De Silva \IEEEmembership{Member, IEEE}
\thanks{A. De Silva, M. V. Perera, N. Wijethilake, N. D. Nanayakkara, A. C. De Silva are with the Department of Electronic and Telecommunication Engineering, University of Moratuwa, Sri Lanka. (email: \{ashwind, malshav, wijethilakemrn.20, nuwan, anjulads\}@uom.lk )}
\thanks{S. Jayasinghe (Professor of Medicine) is with the Department of Clinical Medicine, Faculty of Medicine, University of Colombo, Sri Lanka. (email: saroj@clinmed.cmb.ac.lk)}}

\maketitle

\begin{abstract}
Systemic diseases such as diabetes, hypertension, atherosclerosis are among the leading causes of annual human mortality rate. It is suggested that retinal and conjunctival vascular tortuosity is a potential biomarker for such systemic diseases. Most importantly, it is observed that the tortuosity depends on the thickness of these vessels. Therefore, selective calculation of tortuosity within specific vessel thicknesses is required depending on the disease being analysed. In this paper, we propose a thickness sensitive vessel extraction framework that is primarily applicable for studies related to retinal and conjunctival vascular tortuosity. The framework uses a Convolutional Neural Network based on the IterNet architecture to obtain probability maps of the entire vasculature. They are then processed by a multi-scale vessel enhancement technique that exploits both fine and coarse structural vascular details of these probability maps in order to extract vessels of specified thicknesses. We evaluated the proposed framework on four datasets including DRIVE and SBVPI, and obtained Matthew's Correlation Coefficient values greater than 0.71 for all the datasets. In addition, the proposed framework was utilized to determine the association of diabetes with retinal and conjunctival vascular tortuosity. We observed that retinal vascular tortuosity (Eccentricity based Tortuosity Index) of the diabetic group was significantly higher ($p < .05$) than that of the non-diabetic group and that conjunctival vascular tortuosity (Total Curvature normalized by Arc Length) of diabetic group was significantly lower ($p < .05$) than that of the non-diabetic group. These observations were in agreement with the literature, strengthening the suitability of the proposed framework.
\end{abstract}

\begin{IEEEkeywords}
Multi-scale Vessel Extraction, Convolutional Neural Networks, Retinal and Conjunctival Vascular Tortuosity, Diabetes
\end{IEEEkeywords}

\section{Introduction}
\label{sec:introduction}

\IEEEPARstart{I}{t} is well known that structural changes of retinal vasculature are markers of diabetes, diabetic retinopathy, nephropathy, aging, genetic disorders and cardiovascular diseases \cite{diabretinopathy,twistedbloodvessel}. Clinical observations suggest that these diseases are linked with vascular tortuosity which reflects the twisted and curved nature of blood vessels\cite{twistedbloodvessel}. In particular, retinal fundus images (Fig.\ref{intro_images} (a)) are used in visualizing  microvasculature non-invasively, which has been widely used to examine the association of vascular tortuosity with diabetes and diabetic retinopathy \cite{diabretinopathy}. In addition to diabetic related research, vascular tortuosity has been used in studies pertaining to cardiovascular diseases \cite{Owen2011RetinalAT}, sickle cell retinopathy \cite{khansari2019relationship} and central vein occlusion \cite{yasuda2015significant}.

\begin{figure}[!b]
\centerline{\includegraphics[width = \columnwidth]{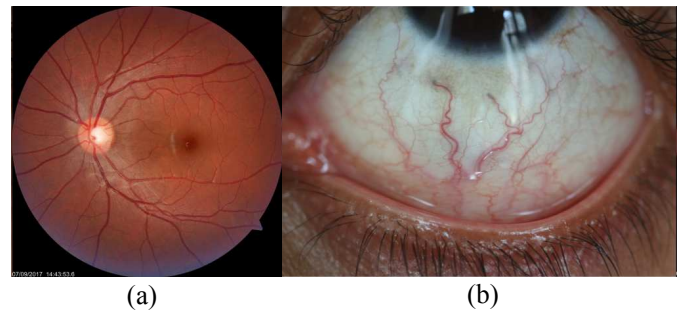}}
\caption{(a) A retinal fundus images (b) An image of the external eye}
\label{intro_images}
\end{figure}

Apart from the retina, the bulbar conjunctiva that covers the sclera of the eye is also a densely vascularized membrane which can be conveniently accessed compared to the retina. Bulbar conjunctival vessels are primarily derived from the ophthalmic artery and could be affected by systemic diseases mentioned above \cite{OpthalmicArtery}. However, to the best of our knowledge, only \cite{Achintha,scleratotuosity} have explored the relationship between diabetes and vascular tortuosity of bulbar conjunctiva. 

Despite the popularity in using retinal fundus images, the acquisition of those is expensive and requires specialized equipment. Unlike the retina, the sclera can be imaged without using expensive specialized equipment. For example, the studies conducted by Iroshan et al. \cite{Achintha} and Sodi et al. \cite{sodi2019quantitative} have used images of the external eye (Fig.\ref{intro_images} (b)) that have been acquired using a regular digital single-lens reflex (DSLR) camera, in order to visualize the bulbar conjunctival vasculature. Therefore, based on the hypothesis that bulbar conjunctival vascular tortuosity acts as a biomarker for systemic diseases, external eye images together with an accurate vessel segmentation algorithm could facilitate large scale patient screening.

Vessel segmentation is a critical step that has to be performed prior to calculating vascular tortuosity in both retinal fundus images and external eye images. Executing this task manually is highly laborious and also impractical due to the high volume of data produced by modern imaging systems. Therefore, numerous automated methods of vessel segmentation have been developed to expedite this task. Traditionally, vessel segmentation from retinal fundus images and external eye images were performed using B-COSFIRE (Bar-Selective Combination of Shifted Filter Responses) filter based algorithms \cite{BCOSFIRE} and morphological operations based algorithms \cite{Morphology}. Recently, Convolutional Neural Networks (CNNs) based methods (UNet\cite{Unet}, DUNET \cite{DUnet}, IterNet \cite{iternet}, R2U-Net \cite{r2unet}, lightweight attention UNet \cite{lightweightattention}) are getting increasingly popular in retinal vessel segmentation tasks, due to their improved performance. Contrary to retinal vessel segmentation, only a few studies have utilized CNNs \cite{ScleraVesselSeg} for segmenting conjunctival vessels. To date, studies \cite{Achintha,Dulara,Morphology} that are related to vascular tortuosity have only used traditional approaches described above. Since CNNs outperform these traditional methods in terms of vessel segmentation accuracy, it is reasonable to assume that vessels segmented using a CNN based method would yield more reliable tortuosity values.

In both retina and bulbar conjuctiva, the tortuous nature of a vessel varies with its thickness \cite{VesselCaliber,sharma2017studies}. Also depending on the study, the required vessel thicknesses could depend on the disease condition. For example, according to \cite{ramos2019computational}, the difference between diabetic and non-diabetic retinal vascular tortuosity is more pronounced in relatively thick vessels. Multi-scale vessel enhancement methods proposed by Frangi et al. \cite{frangi}, Sato et al. \cite{sato19973d} , Steger et al. \cite{steger1998} are popular image processing techniques for enhancing tubular structures such as vessels, while being sensitive towards the structural thicknesses. To the best of our knowledge, only Owen et al. \cite{Owen2011RetinalAT} have employed Steger et al.'s method \cite{steger1998} to measure the vessel thicknesses, but, it does not focus on extracting vessels based on their thicknesses. The methods mentioned above do not contain any CNNs which clearly outperform traditional methods of vessel segmentation. While CNNs has the potential to perform such a task on their own, they would require multiple datasets with explicitly annotated vessels of specified thicknesses, to train multiple models. Constructing multiple datasets and training multiple models for different thicknesses would be tedious and computationally expensive. Therefore, we speculate that a CNN trained on a single dataset with annotations of the entire vasculature paired with a multi-scale vessel enhancement method could boost the accuracy of extracting vessels of specified thicknesses, while eliminating the need to train multiple models on multiple datasets. Nevertheless, there have not been any previous work that had attempted to combine a CNN with a multi-scale vessel enhancement method.

In this work, we propose a novel framework that can extract vessels of specified thicknesses from retinal and external eye images. Such a framework would be highly beneficial for studies involving the computation of retinal and conjunctival vascular tortuosity. In order to extract retinal and conjunctival vessels, we use an IterNet followed by a multi-scale vessel enhancement method that exploits fine and coarse vascular structural details. IterNet has the state-of-the-art performance for retinal vessel segmentation at the time of this study and we hypothesize that it would also yield a better performance in conjunctival vessel segmentation. In the context of the proposed framework, IterNet is used to generate probability maps of the entire vasculature. The multi-scale vessel enhancement method that follows the IterNet, ensures that only the vessels of specified thicknesses are extracted. Therefore, our framework combines the power of CNNs and multi-scale vessel enhancement methods to automatically and accurately extract the vessels of specified thicknesses without requiring multiple datasets or training multiple models. In the case of conjunctival vessels, prior to vessel extraction, a U-Net is used to accurately segment the scleral region from the external eye images.

In addition, we applied our proposed framework to determine the association of diabetes with tortuosity of relatively thick retinal and conjunctival vessels. Our findings from this study agreed with the existing literature, thus fortifying the applicability of the proposed framework on studies related to retinal and conjunctival vascular tortuosity.

\section{METHODOLOGY}

A detailed block diagram of the proposed framework including results at intermediate steps is illustrated in Fig. \ref{framework}. In this section, we describe segmentation of the scleral region (section \ref{ScleraUNetsection} and step A in Fig. \ref{framework}), generation of vessel probability maps (section \ref{generation-of-vessel-prob-maps} and step B in Fig. \ref{framework}), extraction of vessels of specified thicknesses (section \ref{post_processing} and steps C, D, E, F, G, H, I in Fig. \ref{framework}) and calculation of tortuosity (section \ref{tor-calc} and step J in Fig. \ref{framework}). 

\subsection{Segmentation of the Scleral Region}
\label{ScleraUNetsection}

\begin{figure}[!bp]
\centerline{\includegraphics[width = \columnwidth]{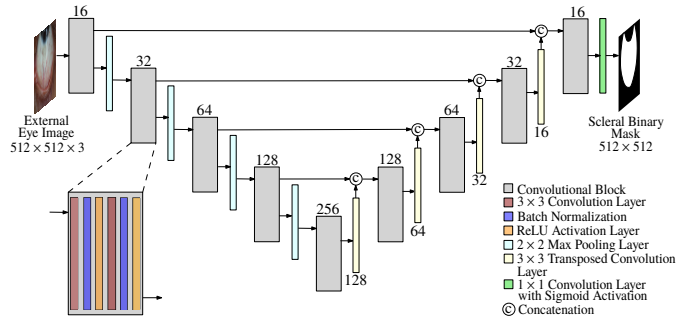}}
\caption{The network architecture of the ScleraUNet. All external eye images are reshaped to $512 \times 512 \times 3$ before feeding them to the network. The number of filters present in each convolutional layer is stated on top of the respective convolutional block. The number of filters are mentioned below each transposed convolutional layer.}
\label{ScleraUNet}
\end{figure}

Unlike retinal fundus images which have an inherent circular region of interest, a pre-processing step has to be performed on external eye images to segment the scleral region before segmenting its conjunctival vessels. Therefore, the accurate segmentation of the sclera from the external eye images is critical for accurate segmentation of conjunctival vessels. In recent years, CNNs have outperformed traditional methods in semantic image segmentation tasks. U-Net \cite{Unet} is a CNN architecture which is widely used in biomedical image segmentation. In this work, we use the U-Net architecture illustrated in Fig. \ref{ScleraUNet}, to segment the scleral region. This network is referred to as \enquote{ScleraUNet} from here onwards to avoid ambiguity.

The ScleraUNet is trained using external eye images and their corresponding annotations of scleral regions. Binary cross-entropy loss is used as the loss function together with ADAM optimizer \cite{adam} to train the network. Random rotations, shifting, flipping, zooming and shearing are used to augment the training data. Details of this training process is described in section \ref{Segmentation of Sclera}.

During inference, the network outputs a binary mask of the scleral region. This binary mask is then used to obtain an image with the isolated scleral region from the external eye image, which we refer to as the scleral image.

\subsection{Generation of Vessel Probability Maps}
\label{generation-of-vessel-prob-maps}

Before extracting the vessels of specified thicknesses from retinal or scleral images, we first generate a vessel probability map that is able to represent the entire underlying vasculature in detail. Each pixel of this map should represent the probability of that pixel belonging to a vessel of the given image. In order to generate these maps, we use the IterNet architecture (Fig. \ref{VesselNet}) proposed by Li et al. \cite{iternet}. IterNet architecture is based on U-Nets and it currently holds state-of-the-art performance in retinal vessel segmentation. Hence, we hypothesize that an appropriately tuned IterNet architecture would also yield a better segmentation of the conjunctival vessels as well.

\begin{figure*}[!bp]
\centerline{\includegraphics[width=\textwidth]{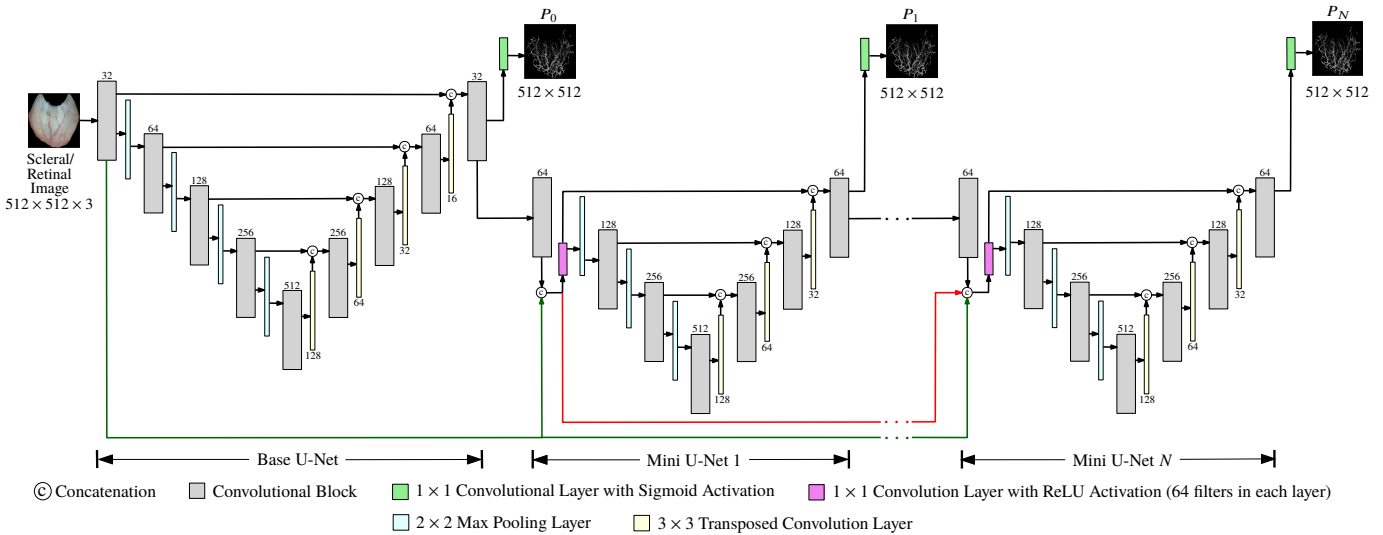}}
\caption{The IterNet architecture; base U-Net and mini U-Nets are connected to each other using a series of connections. Black arrows in the figure represent the skip connections inside each U-Net. Green arrows indicate that the output from the first convolutional block of the base U-Net is sent to each of the mini U-Nets. Red arrows represent the passing of concatenated output from first convolutional block of base U-Net and first mini U-Net, to the remaining mini U-Nets in the network. Note that the convolutional blocks here are identical to the ones that are illustrated in Fig. \ref{ScleraUNet}}.
\label{VesselNet}
\end{figure*}

IterNet contains a base U-Net followed by $N$ number of mini U-Nets. The base U-Net architecture is similar to the U-Net described in section \ref{ScleraUNetsection}. Each mini U-Net is a simplified and light weight version of the base U-Net containing only 8 convolutional blocks. Base U-Net and mini U-Nets, each outputs a probability map where each pixel represents the probability of it belonging to a vessel. The base U-Net outputs a coarse probability map ($P_0$) of the underlying vasculature of a retinal or a scleral image. The subsequent mini U-Nets act as a cascaded unit that refines $P_0$ and outputs probability maps $P_1, P_2, \dots, P_N$. Since the final mini U-Net outputs the most refined probability map, $P_N$ is able to represent the underlying vasculature in detail. Therefore, $P_N$ is considered as the final vessel probability map.

An IterNet model trained with retinal images and their corresponding vessel annotations, is used as the retinal vessel probability map generator. Similarly, an IterNet model trained with scleral images and their vessel annotations, is used as the conjunctival vessel probability map generator.

The total loss $\mathcal{L}$ of the network is computed as the weighted sum of binary cross-entropy losses $\mathcal{L}_i$ $(i = 0, 1, \dots, N)$ ,
\begin{equation}
 \mathcal{L} = \sum_{i=0}^N \theta_i \mathcal{L}_i
\end{equation}
where $\mathcal{L}_0$ is the loss of base U-Net output, $\mathcal{L}_j$ $(j = 1, 2, \dots, N)$ is the loss of $j$\textsuperscript{th} mini U-Net output and $\theta_i$ is the corresponding weight of loss $\mathcal{L}_i$. ADAM optimizer is used to train the network. Random rotations, flipping, zooming, shifting, intensity changes and contrast changes are used to augment the training data. The training process of the IterNet is described in section \ref{VesselNet_results}.

After vessel probability maps are obtained, they are processed together with the original scleral/retinal images to extract the vessels of specified thicknesses, using a multi-scale vessel enhancement method.

\subsection{Extraction of Vessels of Specified Thicknesses}
\label{post_processing}

To extract vessels of specified thicknesses, we use a series of steps that takes in the original scleral/retinal image and its vessel probability map, and returns a binary image that contains the vessels of specified thicknesses as output. For this purpose, we employ a multi-scale vessel enhancement method based on the Frangi filter \cite{frangi} which is specialized in enhancing and separating nearby elongated structures of multiple scales \cite{drechsler2010comparison}. 


\subsubsection{Frangi Filter}
\label{frangi}

The Frangi filter outputs a map with enhanced elongated structures (in this case, vessels) using the multi-scale second order local structure (Hessian) of a given image. Hessian of a grayscale image $I$ at scale $\sigma$ for a pixel $\bm{x} = [x_1, x_2]^T$; $\bm{x} \in \Omega$ where $\Omega$ is the set of all pixels in $I$, is represented by a  matrix $\mathbf{H}(\bm{x}, \sigma) = \big\{h_{ij}(\bm{x}, \sigma)\big\}_{2 \times 2}$ defined as,
\begin{equation}
    h_{ij}(\bm{x}, \sigma) = \sigma^2 I(\bm{x}) * \frac{\partial^2 G(\bm{x}, \sigma)}{\partial x_i \partial
    x_j}
\end{equation}
where $*$ represents the convolution operation and $G(\bm{x}, \sigma) = (2 \pi \sigma^2)^{-1} \exp(-\bm{x}^T \bm{x} / 2 \sigma^2)$ is the bivariate Gaussian.

Let $\lambda_{\sigma, 1}$ and $\lambda_{\sigma, 2}$ be the eigenvalues of $\mathbf{H}(\bm{x}, \sigma)$. With $\lambda_{\sigma, 1}$, $\lambda_{\sigma, 2}$, the Frangi filter response of $I$ at scale $\sigma$ for the pixel $\bm{x}$ is given by the point-wise operator $\mathcal{F}: \Omega \times \mathbb{R} \rightarrow \mathbb{R}$ defined below.

\begin{equation}
\forall \bm{x} \in \Omega : \mathcal{F}(I)(\bm{x}, \sigma) = 
\begin{cases}
    0 &  \lambda_{\sigma, 2} > 0 \\
    V(\mathbf{x}, \sigma)  &  \lambda_{\sigma, 2} \leq 0
\end{cases} 
\label{Frangi-Eqn}
\end{equation}
where, 
\begin{equation}
    V(\bm{x}, \sigma) = \exp\bigg(-\frac{R_B^2}{2\beta^2}\bigg) \bigg(1 - \exp( - \frac{S^2}{2c^2} ) \bigg)
\end{equation}

where, $R_B = |\lambda_{\sigma, 1} / \lambda_{\sigma, 2}|$ and $S^2 = \lambda_{\sigma, 1}^2 + \lambda_{\sigma, 2}^2$. $R_B$ is the blobness measure and $S$ is the second order structuredness. $\beta$ and $c$ are thresholds that control the sensitivity of line filters to the measures $R_B$ and $S$ respectively. For this study, we set $\beta$ and $c$ as: $\beta = 0.5$ and $c = 0.5 \times \displaystyle \max_{\bm{x} \in \Omega} ||\mathbf{H}(\bm{x}, \sigma)||_2$. In this paper, Frangi filter response of $I$ at scale $\sigma$ is denoted as $\mathcal{F}(I)(\sigma)$.

It is important to note that the scale $\sigma$ is related to the vessel thickness $w$ (in pixels) according to $\sigma = Cw$ (see Theorem \ref{theorem:1} of appendix). Here,
\begin{equation}
    C = \frac{1}{2} \sqrt{ 1 - 2W \bigg( \frac{1}{2} \exp \Big( \frac{1}{2} + \log(\alpha) \Big) \bigg)}^{-1}
\end{equation}
where, $W$ is the Lambert W function and $\alpha$ $(0 \leq \alpha <  1)$ is the surround size ratio of the second derivative of the Gaussian kernel. When $\alpha$ is closer to 1, the response of the Frangi filter for vessels becomes sharper. With the relationship $\sigma = Cw$, it follows that $\mathcal{F}(I)(Cw)$ provides a map in which the vessels of thickness $w$ in image $I$ are enhanced.

\subsubsection{Computing Vessel Redness Maps}
We consider the green channel of the original scleral/retinal image to compute the vessel redness map. This is because, the contrast of blood vessels are higher in the green channel of the image. In order to further enhance the vessel contrast, a Contrast Limited Adaptive Histogram Equalization (CLAHE) \cite{clahe} is applied to the green channel image. This contrast enhanced image is then inverted to represent the blood vessels as bright structures in a dark background. Let $I_g$ be this inverted image.

Now we are interested in enhancing vessels of a certain thickness $w$, in $I_g$. For this purpose we can make use of the Frangi filter described in section \ref{frangi}. Applying Frangi filter on $I_g$ together with a gamma correction ($\gamma_r$), yields the vessel redness map $V_r(w)$.
\begin{equation}
   \forall \bm{x} \in \Omega : V_r(\bm{x}, w) = \mathcal{F}(I_g)(\bm{x}, Cw)^{\gamma_r}
\end{equation}
$V_r(w)$ has a stronger emphasis on the degree of vessel redness, but it only provides a coarser enhancement (Fig. \ref{framework} (f)) for vessels of a specified thickness and often unable to preserve fine structural details of vessels. Therefore, we compute a separate vessel map that would have a strong emphasis on the vessel structure.

\subsubsection{Computing Vessel Structural Maps}

Vessel structural map is derived using the final mini U-Net output $P_N$ of the IterNet due to its detailed representation of underlying vasculature. Applying Frangi filter on $P_N$ together with a gamma correction ($\gamma_s$), yields desired vessel structural map $V_s(w)$.
\begin{equation}
    \forall \bm{x} \in \Omega : V_s(\bm{x},w) =  \mathcal{F}(P_N)(\bm{x}, Cw)^{\gamma_s}
\end{equation}
$V_s(w)$ inherits fine structural details of vessels in $P_N$ to provide a finer enhancement (Fig. \ref{framework} (e)) to the vessels of specified thickness $w$.

\subsubsection{Combining Vessel Structural and Redness Maps}

Coarser vessel redness map and finer vessel structural map are combined to yield a combined vessel map $V_c(w)$ as follows. 
\begin{equation}
    \forall \bm{x} \in \Omega : V_c(\bm{x}, w) = \big(V_r(\bm{x}, w) \times V_s(\bm{x}, w)\big)^{\gamma_c}
\end{equation}
where $\gamma_c$ stands for gamma correction. The multiplication ensures that only vessels of specified thickness $w$ in both vessel structural and redness maps would have a higher response. This way, $V_c(w)$ provides a map where vessels of thickness $w$ are enhanced by combining both coarse and fine structural details of vessels. A combined vessel map at a selected thickness $w$ is illustrated in Fig. \ref{framework} (g).

In order to arrive at a map that includes enhanced vessels of specified thicknesses $\mathcal{W} = \{w_i \mid i = 1, 2, 3 \dots, |\mathcal{W}|\}$, we obtain $V_c(w_i)$ for each $w_i \in \mathcal{W}$ and compute $\overline{V_c}(\bm{x}, \mathcal{W})$ as follows.
\begin{equation}
    \forall \bm{x} \in \Omega : \overline{V_c}(\bm{x}, \mathcal{W}) = \max_{w_i \in \mathcal{W}} V_c(\bm{x}, w_i)
    \label{mean-vessel-map-eqn}
\end{equation}
where $\overline{V_c}(\mathcal{W})$ is a map with enhanced vessels of thicknesses $\mathcal{W}$. We then binarize $\overline{V_c}(\mathcal{W})$, with a global histogram threshold computed using Otsu's method \cite{jianzhuang1991automatic}. An example of $\overline{V_c}(\mathcal{W})$ and its binarized version are illustrated in Fig. \ref{framework} (h) and Fig. \ref{framework} (i) respectively.

\subsubsection{Post-processing}
\label{extracting-vessels}

The binarized $\overline{V_c}(\mathcal{W})$ contains small white spots, holes and noisy fragments in addition to the desired vessels of interest. Morphological operations are used to remove these white spots and fill the holes in the binarized $\overline{V_c}(\mathcal{W})$. Then, to remove noisy fragments, we use a connected component based denoising method described below.

First, a set $\mathcal{K}$ containing connected components of the binarized $\overline{V_c}(\mathcal{W})$ is found using an iterative flood-fill algorithm. Then we compute normalized difference ($d_k$) between $k$\textsuperscript{th} component size ($N_k$) and median component size ($M$) as follows.
\begin{equation}
    d_k = \frac{N_k - M}{\displaystyle \max_{k \in \mathcal{K}} |N_k - M|}
\end{equation}
Higher the $d_k$, larger the size of $k$\textsuperscript{th} component. A threshold $t$ was imposed to select the largest connected components (in this case, vessels of interest) based on $d_k$ values. With this selection, we can remove noisy fragments and retain a binary map $V(\mathcal{W})$ that contains vessels of specified thicknesses.

Another form of noise, inherent to the Frangi filter response, arises when there exist vessels in retinal/scleral images with thicknesses greater than the specified thicknesses $\mathcal{W}$. Let $\mathcal{W} = \{w_n, w_{n+1}, \dots, w_{m}\}$ and $w_l$ be the highest vessel thickness in the given retinal/scleral image. If $w_m < w_l$, traces from vessels of thicknesses $\mathcal{U} = \{w_{m+1}, w_{m+2}, \dots, w_{l}\}$ appear in $V(\mathcal{W})$. These traces in $V(\mathcal{W})$ are removed using the following Boolean operation.
\begin{equation}
    \forall \bm{x} \in \Omega: V(\bm{x}, \mathcal{W}) \gets V(\bm{x},\mathcal{W}) \wedge (\sim V(\bm{x},\mathcal{U}))
\end{equation}
\begin{figure*}[!pt]
\centerline{\includegraphics[width=\textwidth]{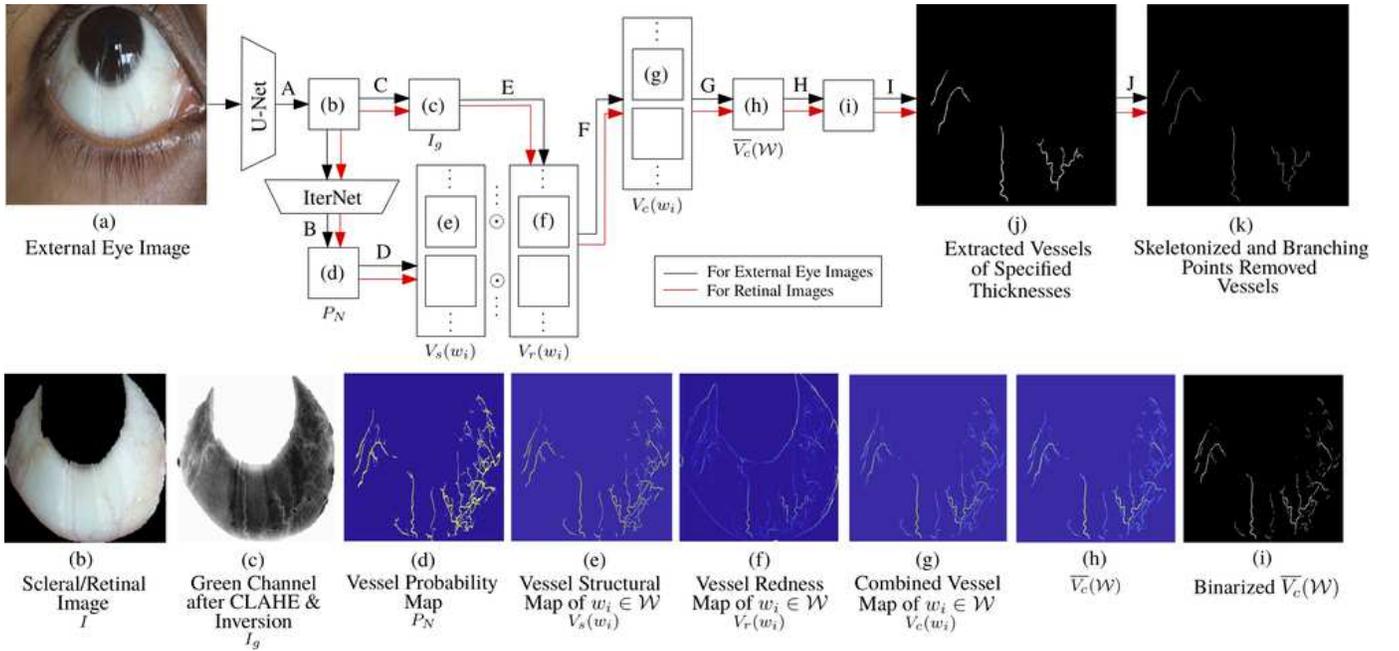}}
\caption{Block diagram of the proposed framework together with a selected external eye image to illustrate each step. The steps of the framework shown above each arrow are as follows: \textbf{A} - Segmenting the scleral region (This step is not applicable for retinal images), \textbf{B} -  Vessel probability map generation, \textbf{C} -  Taking green channel of the scleral/retinal image and applying CLAHE followed by inversion, \textbf{D} -  Computing vessel structural maps for $\forall w_i \in W$ using the vessel probability map, \textbf{E} -  Computing vessel redness maps for $\forall w_i \in W$ using $I_g$, \textbf{F} - Combining vessel structural and redness maps to compute combined vessel maps for $\forall w_i \in W$, \textbf{G} - Obtaining $\overline{V_c}(\mathcal{W})$, \textbf{H} - Binarizing $\overline{V_c}(\mathcal{W})$ with Otsu's method, \textbf{I} - Selecting the largest connected components and noise removal, \textbf{J} - Skeletonizing the vessels followed by branching point removal.}
\label{framework}
\end{figure*}

\subsection{Tortuosity Calculation}
\label{tor-calc}

After extracting vessels from retinal/scleral images, first we skeletonize these vessels to reduce their thickness down to a single pixel. Next, branching points of these skeletonized vessels are detected and removed. In this context, a branching point is defined as a pixel surrounded by more than two 8-connected pixels. With the removal of branching points, we obtain an image where all the sub vessels are separated from each other. Then the tortuosity values of these sub vessels are computed using all the tortuosity indices described in \cite{Dulara} and \cite{tortuosityReview} .  

\section{Experiments and Results}

\subsection{Dataset Description}
\label{datasets}

\subsubsection{SBVPI Dataset}  
Sclera Blood Vessels, Periocular and Iris (SBVPI) is a publicly available dataset primarily intended for sclera and periocular recognition research \cite{SBVPI1,SBVPI2}. It consists of 1840 external eye images ($3000 \times 1700$) in 4 different gaze directions (straight, left, right and up) belonging to 55 healthy subjects. Out of these, the entire conjunctival vasculature is annotated in 128 images and the scleral region is annotated in 1000 images.


\subsubsection{DRIVE Dataset}
Digital Retinal Images for Vessel Extraction (DRIVE) \cite{DRIVE} is a publicly available dataset with 40 retinal images ($565 \times 584$) belonging to 40 subjects (33 healthy subjects and 7 subjects with mild early diabetic retinopathy) along with their corresponding annotations of the entire retinal vasculature. 


\subsubsection{REIDA Dataset}
Retinal and Eye Images for Diabetic Analysis (REIDA) dataset is comprised of 58 retinal fundus images ($3608 \times 3608$) and 60 external eye images ($5184 \times 3456$) that are collected from 32 volunteer subjects of ages 40-70 years, from the National Diabetes Center, Colombo, Sri Lanka. Out of them, 14 are diabetic and 18 are non-diabetic. Subjects with fasting blood glucose values of more than $110 \text{ mgdL}^{-1}$ are categorized as diabetic and the rest as non-diabetic. Retinal fundus images from both eyes of each subject are captured using ZEISS VISUCAM 524 retinal camera. External eye images containing superior and inferior bulbar conjunctiva from both eyes of each subject are captured using a Canon digital single-lens reflex (DSLR) camera with a 100 mm macro lens with no special lighting conditions. We shall denote the retinal image collection as REIDA Retina (REIDA-R) dataset and the external eye image collection as REIDA External Eye (REIDA-EE) dataset. The acquisition procedure of REIDA dataset has been approved by the Ethics Review Committee, Faculty of Medicine, University of Colombo, Sri Lanka (Reference: EC-17-132).  

The REIDA-EE dataset contains scleral annotations of 40 external eye images. In addition, we manually annotated vessels of specified thicknesses in 8 images from each SBVPI, DRIVE, REIDA-R and REIDA-EE datasets. These were used as the test sets to evaluate the performance of extracting vessels of specified thicknesses. Details of these vessel annotations are given in Table \ref{hyperparams}. All of these annotations were verified by an expert.

We evaluated the performance of our proposed framework on the aforementioned datasets. The proposed framework was evaluated in terms of scleral segmentation, vessel probability map generation and extraction of vessels of specified thicknesses. In addition, we applied our framework on REIDA dataset to determine the association of retinal and conjunctival vascular tortuosity with diabetes.

\subsection{Segmentation of Sclera}
\label{Segmentation of Sclera}

The ScleraUnet was evaluated on SBVPI and REIDA-EE where separate models were trained for each dataset. For training, validating and testing the model for SBVPI, 1229, 388 and 223 images were used respectively, whereas for REIDA-EE, 25, 5 and 10 images were used respectively. Both of these models were trained at a learning rate of 0.0001 on an NVIDIA Tesla K80 GPU.

To evaluate the performance of the ScleraUNet, we computed the accuracy score ($Acc$) and Dice Similarity Coefficient ($DSC$). These metrics are defined as follows; $Acc = (TP + TN) / (TP + TN + FP + FN)$ and $DSC = (2 \times TP)/(2 \times TP + FP + FN)$. Here $TP, TN, FP$ and $FN$ indicate true positives, true negatives, false positives and false negatives respectively. For SBVPI, we obtained  $Acc = 0.9868$ and $DSC = 0.9868$, and for REIDA-EE, we obtained $Acc = 0.9701$ and $DSC = 0.9697$.

\begin{figure}[!hb]
\centerline{\includegraphics[width=\columnwidth]{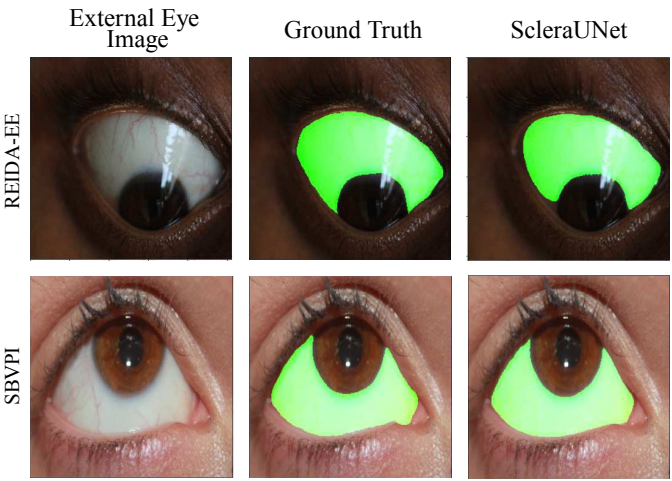}}
\caption{Scleral segmentation results for REIDA-EE (ACC = 0.9779, MCC = 0.9777) and SBVPI (ACC = 0.9850, DSC = 0.9849) datasets}
\label{sclera-seg-results}
\end{figure}

The above accuracy scores and the Dice similarity scores indicate that the ScleraUNet provides sufficient performance for this task. A selected external eye image from each SBVPI and REIDA-EE, together with their scleral ground truth and scleral segmentations obtained from ScleraUNet are shown in Fig. \ref{sclera-seg-results}.

\subsection{Vessel Probability Map Generation}
\label{VesselNet_results}

We configured IterNet to output retinal vessel probability maps as follows; the number of mini U-Nets of the architecture was set to $N=1$ with loss weights $\theta_0 = 1$ and $\theta_1 = 0.2$. We have denoted this IterNet architecture as R-IterNet (R - representing retina). To train R-IterNet, we used retinal images and their corresponding annotated vessels from the DRIVE dataset. The network was trained using 29 images, validated on 3 images and tested with 5 retinal images. 

To output conjunctival vessel probability maps, we configured IterNet as follows; the number of mini U-Nets of the network architecture was set to $N=2$ with loss weights $\theta_0 = \theta_1 = \theta_2 = 1$. We have denoted this IterNet architecture as C-IterNet (C - representing conjunctiva). This network was trained on scleral images and their corresponding vessel annotations from the SBVPI dataset. The network was trained using 108 images, validated and tested on 10 images each.

The number of mini U-Nets and the loss weights of both R-IterNet and C-IterNet were determined using a grid search for combinations of $N = \{0, 1, \dots, 4\}$ and $\theta_0 = \theta_1 = \theta_2 = \{0, 0.1, 0.2, \dots, 1\}$. Both networks were trained at a learning rate of 0.0001 on an NVIDIA Tesla K80 GPU.

During evaluation, generated probability maps were binarized using a threshold of 0.5, before computing the accuracy. The trained R-IterNet and C-IterNet models yielded testing accuracies of 96.74\% and 95.65\% respectively, on their corresponding test sets. 

For this study, we did not train or evaluate R-IterNet and C-IterNet using REIDA datasets, due to the unavailability of annotations of the entire vasculature. Based on the hypothesis that retinal and scleral images from REIDA dataset were characteristically similar to the images from DRIVE and SBVPI datasets respectively, we used R-IterNet (trained on DRIVE) and C-IterNet (trained on SBVPI) models to generate vessel probability maps for REIDA images. Vessel probability maps of selected images from each dataset are shown in Fig. \ref{montage}.

\begin{figure*}[!bp]
\centerline{\includegraphics[width=\textwidth]{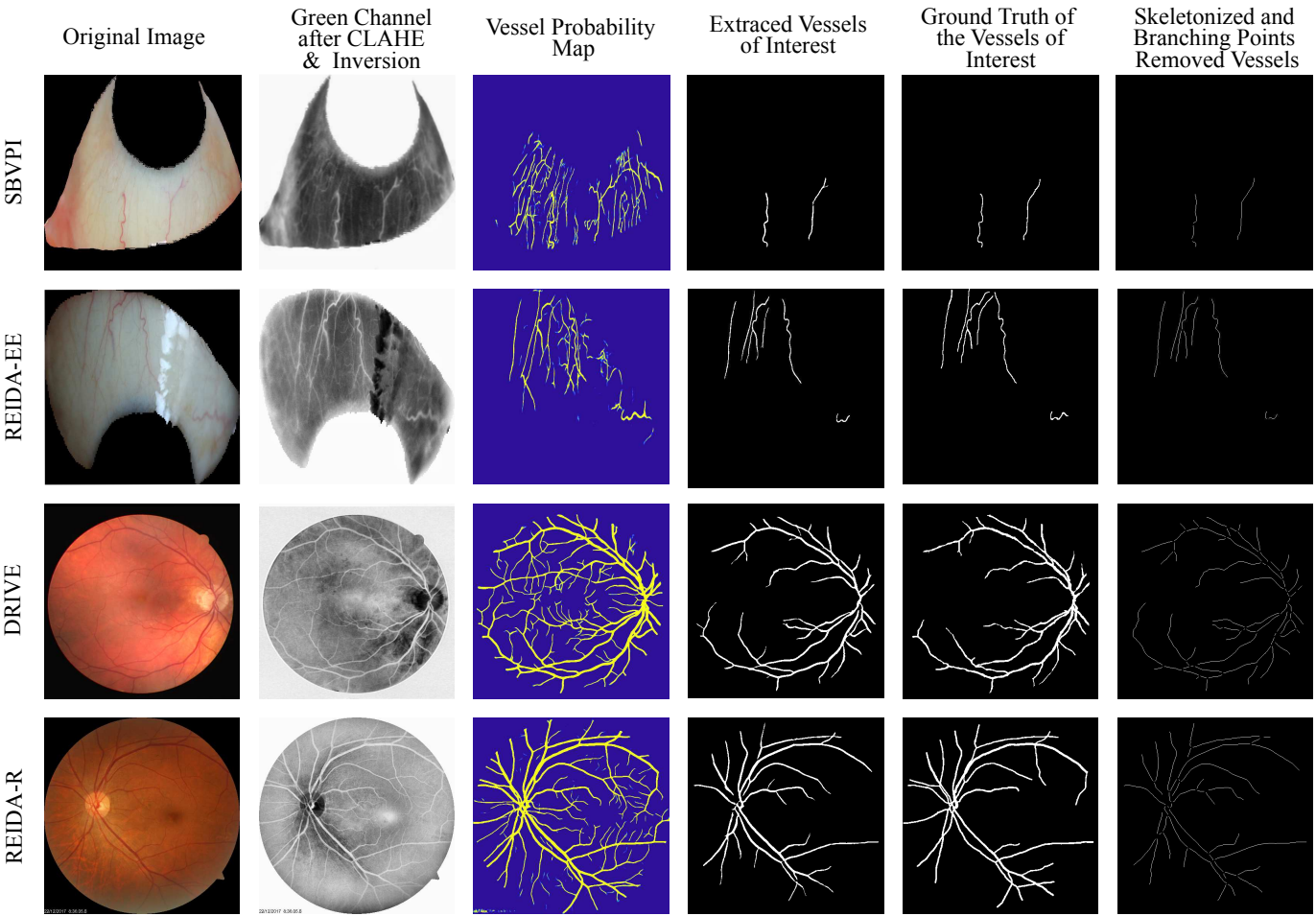}}
\caption{Resultant of selected images from each dataset at different steps of the proposed framework for SBVPI (MCC = 0.8060), REIDA-EE (MCC = 0.7432), DRIVE (MCC = 0.8147) and REIDA-R (MCC = 0.7481) datasets. Note that we have not preserved the original aspect ratios of these images for the purpose of illustration.}
\label{montage}
\end{figure*}

\subsection{Extraction of Vessels of Specified Thicknesses}

To evaluate the performance of extracting vessels of specified thicknesses, we used 8 images from each dataset (DRIVE, SBVPI, REIDA-R and REIDA-EE) together with their corresponding annotations as mentioned in section \ref{datasets}. Thicknesses ($\mathcal{W}$) of the annotations in each dataset and hyperparameters ($\gamma_r, \gamma_s, \gamma_c$ and $t$) used for extracting these vessels are tabulated in Table \ref{hyperparams}. Note that these thickness values are given for retinal/scleral images that are resized to $512 \times 512$ while preserving the aspect ratio in order to avoid distortions in vessel structure. Values for all hyperparameters except for $\mathcal{W}$ were set using a grid search for combinations of $\gamma_r = \gamma_s = \gamma_c = \{0, 0.1, \dots, 1\}$ and $t = \{0.01, 0.02, \dots, 0.50\}$. As described in section \ref{frangi}, since a value closer to 1 would result in a sharper response, $\alpha$ was empirically set to 0.9. $\mathcal{W}$ was set based on the specified thicknesses (in pixels) of vessels. For example, to extract vessels of thicknesses ranging from 4 to 8 pixels, $\mathcal{W}$ was set to $\{4, 5, 6, 7, 8\}$. In Fig. \ref{montage}, we illustrate selected resultant images from each dataset at given intermediate steps of the proposed framework.

\begin{table}[t]
\caption{Hyperparameters used for extracting vessels of specified thicknesses in each dataset}
\footnotesize
\begin{center}
\begin{tabular*}{\columnwidth}{l|c|c|c|c}
\hline \hline
Dataset & S.T$^{*}$ (in pixels) & $\mathcal{W}$ & $\gamma_r$, $\gamma_s$, $\gamma_c$ & $t$\\
\hline
DRIVE & 7 to 12 & \{7, 8, \dots , 12\} & 0.4, 0.7, 0.8 & 0.05 \\ 
SBVPI & 4 to 8 & \{4, 5, 6, 7, 8\} & 0.7, 0.1, 0.9 & 0.3 \\
REIDA-R &  7 to 12 & \{7, 8, \dots , 12\} & 0.9, 0.4, 0.5 & 0.05 \\
REIDA-EE & 4 to 8 & \{4, 5, 6, 7, 8\} & 0.7, 0.7, 0.7 & 0.2 \\\hline \hline
\end{tabular*}
\vspace{2pt} \\

{\raggedright $^{*}$S.T stands for specified thicknesses of vessels that have to be extracted. \par}
\label{hyperparams}
\end{center}
\end{table}

The extracted vessels consist of only a smaller number of foreground pixels compared to background pixels, therefore, resulting in a significant class imbalance. Hence, apart from the accuracy, we use Matthew’s Correlation Coefficient ($MCC$) given in \eqref{eq5} which is widely used as an appropriate evaluation metric for class imbalanced binary classification problems \cite{chicco2020advantages}.
\begin{equation}
MCC = \frac{TP/N - (S\times P)}{\sqrt{P \times S \times (1- S) \times (1- P)}}\label{eq5}
\end{equation}
Where $S = (TP + FN)/N$, $P = (TP + FP)/N$ and $N = TP + FP + TN + FN$.
In Table \ref{prom_vessel_results}, we report the performance metrics obtained for extracting vessels of specified thicknesses with respect to each dataset. In the same table, we also present results for the following cases of computing $\overline{V_c}(\mathcal{W})$: 

\begin{enumerate}[label=(\roman*)]
    \item Using only vessel redness maps ($V_r$) :\\ $\overline{V_c}(\bm{x},\mathcal{W}) = \displaystyle \max_{w_i \in \mathcal{W}} V_r(\bm{x}, w_i)$
    \item Using only vessel structural maps ($V_s$) :\\ $\overline{V_c}(\bm{x},\mathcal{W})  = \displaystyle \max_{w_i \in \mathcal{W}} V_s(\bm{x}, w_i)$
    \item  Using the combined vessel maps ($V_c$) :\\ $\overline{V_c}(\bm{x},\mathcal{W}) = \displaystyle \max_{w_i \in \mathcal{W}} V_c(\bm{x}, w_i)$
\end{enumerate}







\begin{table}[!h]
\caption{Performance Evaluation of Extracting Vessels of Specified Thicknesses}
\small
\begin{center}
\begin{tabular*}{\columnwidth}{p{48pt}|>{\centering\arraybackslash}p{23pt}|>{\centering\arraybackslash}p{65pt}|>{\centering\arraybackslash}p{65pt}}
\hline
\hline
Dataset & Cases & \textbf{$Acc$} &\textbf{$MCC$} \\

\hline
      & $V_r$ & 0.9828 $\pm$ 0.0059 &  0.3925 $\pm$ 0.1836 \\
 SBVPI& $V_s$ & 0.9869 $\pm$ 0.0065&  0.4910 $\pm$ 0.1506 \\
      & $V_c$ & \textbf{0.9969 $\pm$ 0.0015}  &  \textbf{0.7335 $\pm$ 0.0536}  \\

\hline 
          & $V_r$ & 0.9786 $\pm$ 0.0037 & 0.4370 $\pm$ 0.0925 \\
REIDA-EE & $V_s$ &0.9901 $\pm$ 0.0043 & 0.5944 $\pm$ 0.0869 \\
          & $V_c$ & \textbf{0.9938 $\pm$ 0.0023}  & \textbf{0.7265 $\pm$ 0.0204}  \\

\hline
      & $V_r$ & 0.9127 $\pm$ 0.0094 & 0.5441 $\pm$ 0.0317  \\
DRIVE & $V_s$ & 0.9685 $\pm$ 0.0041  & 0.7547 $\pm$ 0.0315 \\
      & $V_c$ & \textbf{0.9763 $\pm$ 0.0016}  & \textbf{0.7953 $\pm$ 0.0104}  \\

\hline
             & $V_r$ & 0.9141 $\pm$ 0.0061& 0.5388 $\pm$ 0.0274  \\
REIDA-R & $V_s$ & 0.9583 $\pm$ 0.0066 & 0.6796 $\pm$ 0.0450 \\
             & $V_c$ & \textbf{0.9673 $\pm$ 0.0037} & \textbf{0.7167 $\pm$ 0.0223} \\

\hline
\hline
\end{tabular*}
\vspace{2pt} \\
\label{prom_vessel_results}
\end{center}
\end{table}

Out of above three cases, it was observed that the highest accuracy score and $MCC$ values were obtained when the combined vessel maps were utilized (case (iii)) to compute $\overline{V_c}(\mathcal{W})$, rather than using only vessel structural maps (case (ii)) or only vessel redness maps (case (i)) individually. Therefore, the combination of vessel structural and vessel redness maps given by \eqref{mean-vessel-map-eqn} is better suited for accurate extraction of vessels of specified thicknesses.

Moreover, we provide qualitative results in Fig. \ref{zoomed-images} which illustrates several closeup visualizations of extracted vessels belonging to different sets of thicknesses other than those mentioned in Table \ref{hyperparams}.

\begin{figure}[!tp]
\centerline{\includegraphics[width=\columnwidth]{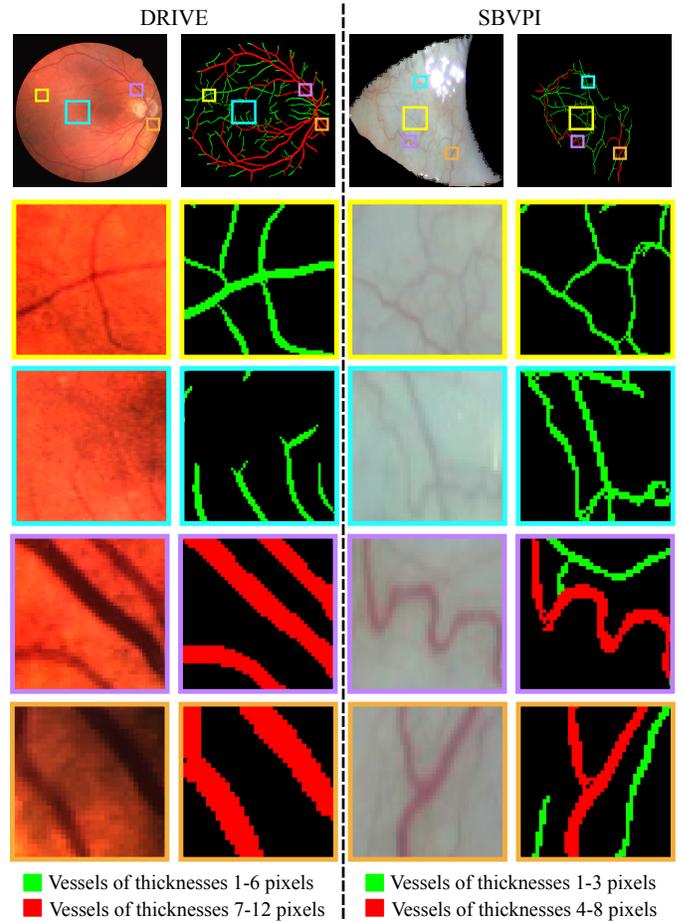}}
\caption{Closeup visualizations of extracted vessels of specified thicknesses from DRIVE and SBVPI. We extracted retinal vessels of thicknesses 1-6 pixels ($\gamma_r = 0, \gamma_s = 0.7, \gamma_c = 0.8, t = 0.05$) and 7-12 pixles ($\gamma_r = 0.4, \gamma_s = 0.7, \gamma_c = 0.8, t = 0.05$) from the DRIVE dataset. Similarly, conjunctival vessels of 1-3 pixels ($\gamma_r = 0, \gamma_s = 0.1, \gamma_c = 0.9, t = 0.3$) and 4-8 pixels ($\gamma_r = 0.7, \gamma_s = 0.1, \gamma_c = 0.9, t = 0.3$) were extracted from the SBVPI dataset. }
\label{zoomed-images}
\end{figure}

\subsection{Determination of the Association of Retinal and Conjunctival Vascular Tortuosity with Diabetes}
\label{diabetic-study}

As a potential application, we used our proposed method on REIDA-R and REIDA-EE datasets to extract vessels of specified thicknesses given in Table \ref{hyperparams}, and attempted to determine the association of retinal and conjunctival vascular tortuosity with diabetes.

The vascular tortuosity values based on a total of 12 tortuosity indices described in \cite{tortuosityReview,Dulara} were computed for retinal and conjunctival vessels of both diabetic and non-diabetic subjects. When a subject was having more than one retinal/external eye image, the mean tortuosity value of those was attributed to that particular subject.

Since sample sizes of diabetic ($n = 14$) and non-diabetic ($n = 18$) groups were not equal and the tortuosity values of the two groups were not distributed normally, we used the unpaired Mann-Whitney test to compare vascular tortuosity values between diabetic and non-diabetic groups. 

This test yielded that retinal vascular tortuosity calculated with Eccentricity based Tortuosity Index (ETI) \cite{Dulara} of the diabetic group (median = $0.0167$) was significantly higher ($U = 74, p = .0247$) than that of the non-diabetic group (median = $0.0150$), and conjunctival vascular tortuosity calculated with Total Curvature normalised by Arc Length (TCAL) \cite{TCAL}  of the diabetic group (median = $0.0772$) is significantly lower ($U = 80, p = .0420$) than that of the non-diabetic group (median = $0.0803$). Fig. \ref{tor_figures} (a) and (b) illustrate the comparison between two groups in terms of ETI and TCAL values for retinal and conjunctival vessels respectively. It was also observed that there are no significant differences between diabetic and non-diabetic groups in terms of the other considered tortuosity indices for both retinal and conjunctival vessels. 

\begin{figure}[!h]
\centerline{\includegraphics[width =1\columnwidth]{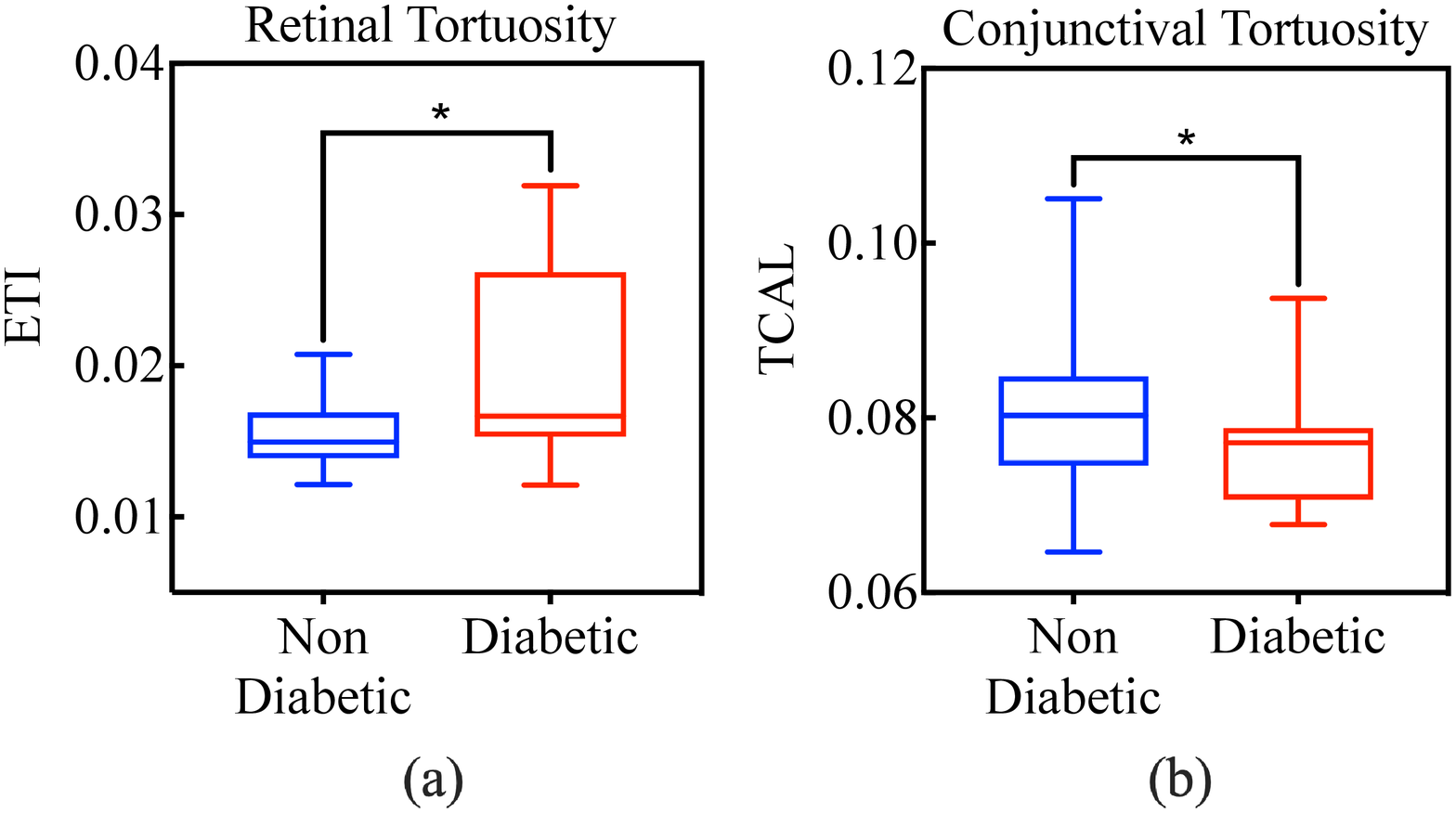}}
\caption{The comparison between diabetic and non-diabetic groups with respect to (a) retinal vascular tortuosity in terms of ETI. (b)  conjunctival vascular tortuosity in terms of TCAL. \enquote{$*$} represents the statistical significance of the difference between two groups as measured by Mann-Whitney test. The horizontal bar in each box plot denotes the median of the respective group. }
\label{tor_figures}
\end{figure}


In addition to comparing the tortuosity values, we also analyzed the relationship between tortuosities of retinal vessels and conjunctival vessels. The correlation between retinal and conjunctival vascular tortuosity values for each tortuosity index was computed separately for non-diabetic and diabetic groups using Spearman's rank correlation. However there were no statistically significant correlation $(p > .05)$ between retinal and conjunctival vascular tortuosity values measured in terms of any of the considered tortuosity indices in both non-diabetic and diabetic groups.

\section{Discussion}

In this study, the proposed framework included an IterNet based CNN to obtain probability maps of the entire retinal/conjunctival vasculature which were then subjected to a series of post-processing steps based on a multi-scale vessel enhancement method, that exploits both fine and coarse structural vascular details of these probability maps in order to extract vessels of specified thicknesses. This way, we could incorporate the power of CNNs into the framework without requiring multiple datasets with explicitly annotated vessels of different thicknesses and training multiple models. The CNNs of this framework were only needed to be trained once, while the vessel extraction steps could be conveniently configured to suit the nature of the vessels of interest.

The proposed framework achieved $MCC$ values of $0.7350, 0.7953, 0.7167$ and $0.7265$ for SBVPI, DRIVE, REIDA-R and REIDA-EE respectively in extracting the vessels of the specified thicknesses given in Table \ref{hyperparams}. This suggests that the framework attains sufficient vessel extraction performance for the selected sets of thicknesses. Since no previous study has introduced a similar framework that focuses on extracting retinal and conjunctival vessels of specified thicknesses, we did not compare the performance of the proposed framework with a previous study. 

In the framework, we employed a U-Net (referred to as ScleraUNet) to perform the segmentation of the scleral region of the external eye images. Accurate scleral region segmentation positively contributes to the overall conjunctival vessel extraction performance of the framework. As evident from section \ref{Segmentation of Sclera}, the ScleraUNet attained nearly perfect segmentation accuracy values for SBVPI ($Acc = 0.9868$) and REIDA-EE ( $Acc = 0.9701$), thus demonstrating its suitability for this task.

Within the proposed framework, we used an IterNet architecture to generate the vessel probability maps for both retinal and conjunctival vessels. Since IterNet currently holds the state-of-the–art performance in retinal vessel segmentation at the time of this study, we hypothesized that a similar appropriately trained network can be employed to generate accurate conjunctival vessel probability maps. The testing segmentation accuracy of $95.65\%$ achieved by the C-IterNet for SBVPI dataset justified the validity of this hypothesis. A major advantage of the proposed framework is the fact that an IterNet has to be trained only once on a single dataset with annotations of the complete vasculature, rather than training it on multiple datasets with explicit vessel annotations of different thicknesses. Thereafter, the vessels of interest are extracted from the vessel probability maps generated by the trained IterNet, using a multi-scale vessel enhancement method that does not require any prior training.

The vessel extraction method of the proposed framework ensures that the vessels of desired sets of thicknesses are extracted with the combined vessel maps ($V_c$) computed using vessel structural maps ($V_s$) and vessel redness maps ($V_r$). Based on results in the Table \ref{prom_vessel_results}, the highest $MCC$ values were obtained for both retinal and conjunctival vessels, when we used these $V_c$ in the framework instead of using $V_s$ or $V_r$ individually. In all the datasets considered in the evaluation, the vessels of interest had pronounced shades of red. Therefore, the combination of vessel redness maps with vessel structural maps provided a better enhancement of vessels of specified thicknesses. However, in an application where vessel redness is less pronounced (for example, when trying to enhance thinner or faded vessels), $\gamma_r$ should be set to a value closer to zero, in order to obtain a better vessel enhancement, as illustrated in Fig. \ref{zoomed-images}. However, due to the unavailability of manual annotations we did not perform a quantitative evaluation for extracting vessel thicknesses lesser than the thicknesses specified in Table \ref{hyperparams}. 

As a potential application, we used the proposed framework to determine the association of retinal and conjunctival vascular tortuosity with diabetes. From the Mann-Whitney tests performed on the retinal tortuosity values as described in section \ref{diabetic-study}, the tortuosity index ETI, showed a significantly higher value  ($p < .05$) for diabetic subjects. This is consistent with the previous study by Ramos et al. \cite{ramos2019computational} in which the tortuosities were calculated with thick vessels. Moreover, the results from our study showed that conjunctival vascular tortuosity calculated with the TCAL index is significantly lower ($p < .05$) in diabetic subjects than that of non-diabetic subjects. This result also agrees with the findings of a study conducted by Owen et al. \cite{scleratotuosity} which observed that the tortuosity of conjunctival macro-vessels are lower in diabetic subjects. 

There were no significant correlation between retinal and conjunctival vascular tortuosities for neither diabetic nor non-diabetic subjects. We speculate that, even though this result indicates that there is no apparent relationship between retinal and conjunctival vascular tortuosities, the lack of representation power of the existing tortuosity indices may have affected the outcome of our correlation study. This may be due to the fact that most of the widely used tortuosity indices in the literature do not take into consideration, the differences between the structural properties pertaining to these retinal and conjunctival vessels. Thus, existing indices might not be suitable for studies that aim to determine the relationship between retinal and conjunctival vascular tortuosities. Hence, a novel tortuosity index that takes into account the differences between the vessel structures would be required in order to perform a better comparison between retinal and conjunctival vascular tortuosity values.

\section{Conclusion}

In this paper, we proposed a novel framework which encompass a CNN paired with a multiscale vessel enhancement method to extract vessels of specified thicknesses from retina and/or conjunctiva which can be applied in studies related to vascular tortuosity. Since the framework achieved $MCC$ values greater than 0.71 for the considered datasets in extracting the vessels of specified thicknesses, it can be concluded that the framework performs well in extracting both retinal and conjunctival vessels of different thicknesses. In addition, we also applied the framework to determine the association of the tortuosity of relatively thick vessels in retina and conjunctiva, with diabetes and found that the obtained tortuosity comparisons were in agreement with the existing literature, thus strengthening the applicability of the proposed framework in vascular tortuosity related studies.


\section{Acknowledgements}

We express our sincere gratitude to Dr. Mahen A. Wijesuriya, Dr. Chamari L. Warnapura and other staff members of the National Diabetes Centre, Rajagiriya, Sri Lanka for the immense support provided for the data collection procedure. Also, we thank Mr. Achintha Iroshan and Mr. Dulara de Zoysa for their support in refining this work.

\appendix[Relationship between Width $w$ and Scale $\sigma$]

The second order derivative of a Gaussian kernel at scale $\sigma$ creates a probe kernel as illustrated in Fig. \ref{theorem}.

\begin{figure}[!h]
\centerline{\includegraphics[width=\columnwidth]{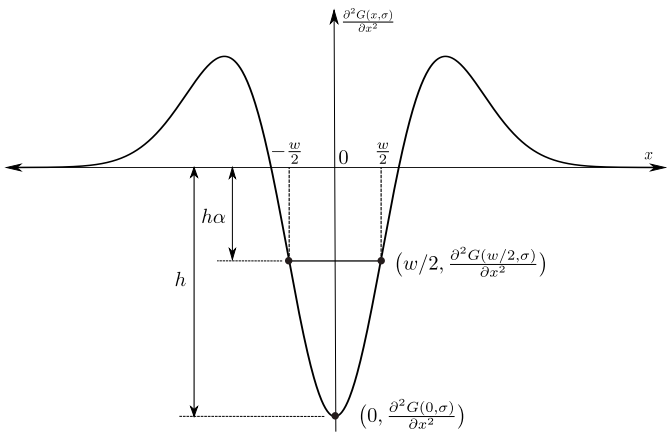}}
\caption{The second order derivative of Gaussian probe kernel}
\label{theorem}
\end{figure}

It is of interest to obtain an analytical relationship between the vessel thickness of interest $w$ and $\sigma$ when vessel thickness is matched with the middle lobe width of the probe kernel at the surround size ratio $\alpha$ ($0 \leq \alpha < 1)$. This analytical relationship is introduced in Theorem \ref{theorem:1}.

\begin{theorem}
    The relationship between the vessel width $w$ and the scale $\sigma$ of the second derivative of Gaussian kernel $\pdv[2]{G(x, \sigma)}{x}$, when 
    \begin{equation}
        \alpha \abs{\pdv[2]{G(0, \sigma)}{x}}  = \abs{\pdv[2]{G(w/2, \sigma)}{x}}, 0 \leq \alpha < 1
        \label{premise}
    \end{equation}
    
    is given by the following equation. 
    
    $$\sigma = \frac{w}{2} \sqrt{ 1 - 2W \bigg( \frac{1}{2} \exp \Big( \frac{1}{2} + \log(\alpha) \Big) \bigg)}^{-1}$$
    
    where $W$ is the Lambert-W function.
    \label{theorem:1}
\end{theorem}

\begin{proof}
Substituting $G(\bm{x}, \sigma) = \sqrt{2 \pi \sigma^2} \exp(-x^2/2\sigma^2)$ in to \eqref{premise}, we arrive at, 
\begin{equation}
    \exp(-w^2/2\sigma^2) \big( -w^2/\sigma^2 + 1 \big) = \alpha
    \label{first}
\end{equation}
Taking the natural logarithm in both sides of \eqref{first}, we get 
\begin{equation}
    \log(-2y + 1) = y + \log(\alpha)
    \label{second}
\end{equation}
where $y = w^2/\sigma^2$. By arranging the terms of \eqref{second}, we arrive at a $az + b\log(z) + c = 0$ type equation where $a = -1/2$, $b = -1$, $c = 1/2 + \log(1-\alpha)$ and $z = -2y + 1 = -2w^2/\sigma^2 + 1$. The general solution of $az + b\log(z) + c = 0$ is 
\begin{equation}
    z = (b/a) W \big( (a/b) \exp (-c/b) \big)
    \label{third}
\end{equation}
where $W$ is the Lambert-W function.

Substituting for $a$, $b$, $c$ and $z$ in \eqref{third} yields the following.
\begin{equation}
    \sigma = \frac{w}{2} \sqrt{ 1 - 2W \bigg( \frac{1}{2} \exp \Big( \frac{1}{2} + \log(\alpha) \Big) \bigg)}^{-1}
\end{equation}
\end{proof}

\bibliographystyle{IEEEtran}
\bibliography{IEEEfull}

\end{document}